\documentclass[11pt]{amsart}
\usepackage{xcolor}
\usepackage[top=15mm, bottom=20mm, left=30mm, right=30mm]{geometry}
\usepackage[style=numeric, isbn=false, doi=false, url=false, backend=bibtex]{biblatex}
\DeclareFieldFormat[article,periodical]{number}{}
\renewbibmacro{in:}{%
\ifentrytype{article}{}{%
\printtext{\bibstring{in}\intitlepunct}}}
\DeclareFieldFormat[article,incollection]{pages}{#1}%

\usepackage{enumerate}
\usepackage{hyperref}

\newcommand{\rats}{\mathbb{Q}}

\newcommand{\reals}{\mathbb{R}}
\newcommand{\preals}{\ensuremath{\reals_+}}

\DeclareMathOperator{\scor}{\mathbf{\Sigma}}
\DeclareMathOperator{\unam}{\mathbf{S}}
\DeclareMathOperator{\sunam}{\mathbf{S}}
\DeclareMathOperator{\wunam}{\mathbf{W}}
\DeclareMathOperator{\maj}{\mathbf{M}}  
\DeclareMathOperator{\cond}{\mathbf{C}}

\DeclareMathOperator{\SP}{\mathbf{SP}}
\newcommand{\cons}{\mathcal{K}}
\newcommand{\profs}{\mathcal{P}}
\newcommand{\votesits}{\mathcal{V}}
\newcommand{\votedists}{\mathcal{P}}
\newcommand{\quot}{\mathcal{Q}}
\newcommand{\elecs}{\mathcal{E}}

\DeclareMathOperator{\simp}{\Delta}

\DeclareMathOperator{\R}{\mathcal{R}}

\def\noproof{\hfill\ensuremath{\square}}
\DeclareMathOperator{\nummap}{\mathcal{N}}
\DeclareMathOperator{\distmap}{\mathcal{D}}
\newcommand{\dt}{\tilde{d}}

\newcommand{\hl}[1]{\textnormal{\textbf{#1}}} 
\newcommand{\comment}[1]{\textcolor{red}{\textit{#1}}}

\theoremstyle{plain}
\newtheorem{defn}{Definition}[section]
\newtheorem{remark}[defn]{Remark}

\newtheorem{prop}[defn]{Proposition}

\newtheorem{cor}[defn]{Corollary}
\newtheorem{eg}[defn]{Example}

\begin{document}

\title{Distance rationalization of social rules}
\author{Benjamin Hadjibeyli}
\address{ENS Lyon}
\email{benjamin.hadjibeyli@ens-lyon.fr}
\author{Mark C. Wilson}
\address{University of Auckland}
\email{mcw@cs.auckland.ac.nz}

\begin{abstract}
The concept of distance rationalizability of social choice rules has been
explored in recent years by several authors. We deal here with several foundational questions, and unify, correct, and generalize previous work. For example, we study a new question involving 
uniqueness of representation in the distance rationalizability framework, and present a counterexample.

 For rules satisfying various axiomatic properties such as anonymity, neutrality and homogeneity, the standard profile representation of input can be 
compressed substantially. We explain in detail using quotient constructions and symmetry groups how distance rationalizability 
is interpreted in this situation. This enables us to connect the theory of distance rationalizability with geometric concepts such as Earth Mover distance and optimal transportation. We expect this connection to prove fruitful in future work.

We improve on the best-known sufficient conditions for  rules rationalized via votewise distances
to satisfy 
anonymity, neutrality, homogeneity, consistency and continuity. This leads to a class of well-behaved rules which deserve closer scrutiny in future.
\end{abstract}


\thanks{We thank John Hillas for raising the uniqueness question in Section~\ref{ss:unique}. The second author thanks Elchanan Mossel and Miklos Racz for useful
conversations.}
\maketitle

\section{Introduction} \label{s:intro}

The number of possible social choice rules is huge, and even the number
of those singled out in the literature for analysis is large.
Researchers have tried many different axioms in order to classify and
characterize these rules, sometimes leading to impossibility theorems.
New rules are still being introduced, and the subject is far from tidy.

A promising unifying framework is that of \emph{distance
rationalization} (abbreviated DR), whereby some subset $D$ (the
\emph{consensus set}) of the set of elections is distinguished. That
subset is further partitioned into  finitely many subsets, on each of which there is a
different social outcome. We choose a distance measure $d$ on elections,
and for each election $E$ outside $D$, an outcome is chosen socially if
and only if it is chosen in some election in $D$ minimizing the distance
under $d$ to $E$. 
This approach
``decomposes" a rule into simpler components $D$ and $d$.

Although the basic idea is quite old (arguably going back to Condorcet's maximum likelihood approach to voting), systematic study of this approach began with
Nitzan, Lerer and Campbell \cite{Nitz1981, LeNi1985, CaNi1986}.  Their work shows that almost 
every known rule can be represented
in the DR framework, and the main interest in the subject is when we can
choose the distance and consensus notions to be natural and computationally tractable. More recently,  Elkind, Faliszewski, and Slinko \cite{EFS2012, EFS2015, ES2015}
 have further developed the theory, following  Meskanen and Nurmi \cite{MeNu2008}.
In particular, they have focused on the important class of \emph{votewise} distances and 
obtained useful sufficient conditions on $(D,d)$ so that the induced rule
satisfies desirable properties such as monotonicity, anonymity, and
homogeneity. 

\subsection{Our contribution} \label{ss:contrib} 
We deal with several foundational and definitional issues,
many of which have not been discussed by previous authors (in some cases
because the level of generality they used was not sufficient to
distinguish these concepts). Some of our
contribution consists of a more efficient and rigorous presentation of
known material. Section~\ref{s:defs} develops the basic notation and terminology. Our
approach is similar to that taken by previous researchers, but there are
some improvements. We aim to operate generally (for example, by using a hemimetric 
rather than a metric, and considering social choice and social welfare functions in a single analysis), and explicitly distinguish several concepts that have
sometimes been conflated in previous work. In Section~\ref{s:foundation} 
we give necessary and sufficient conditions for a rule to be distance rationalizable, improving slightly on results of the abovementioned authors. We pose an interesting question regarding uniqueness of representation
in the DR framework, which does not appear to have been noticed before. We give 
a counterexample in Section~\ref{ss:unique}.

In Section~\ref{s:quot} we explain how 
equivalence relations and symmetries between elections allow us to describe DR 
rules more compactly, and make the connection between the original profile-based definitions and the quotient representations explicit. The distinction between compatible and totally compatible distances is important and new, and the idea of a distance being \emph{simple} with respect to an equivalence relation is also new as far as we know.  None of our results in this section rely on the distance being votewise and are proved for general consensuses; the applications therefore generalize results of Elkind, Faliszewski and Slinko \cite{EFS2015}. In particular, we make the connection between $\ell^1$-votewise distances and 
the Earth Mover distance, relating the subject of distance rationalization of anonymous rules to the theory of optimal transportation and maximum weight matchings. We believe that this new connection will prove 
fruitful in future work.

We apply the above results to neutrality and anonymity, obtaining complete characterizations in Propositions~\ref{prop:neut} and~\ref{prop:anon}. In Section~\ref{s:homog} we deal with homogeneity, which is not quite covered by the results on groups. Our approach shows that the reason Dodgson's rule is not homogeneous is because the equivalence relation is induced by the action of a monoid (``group without inverses") that is not a group. 

Specializing to votewise distances, we concentrate in Section~\ref{s:VMP} on what we term the Votewise Minimizer Property, which is a way of requiring the consensus and distance to combine well. This allows us to give  improved sufficient conditions for DR rules to satisfy homogeneity, consistency, and continuity.

\section{Basic definitions} \label{s:defs}

We use standard concepts of social choice theory. Not all of these concepts have completely 
standardized names. We shall need to deal with several candidate and voter sets simultaneously, which explains the generality of our definitions.
However in many cases  it suffices to deal with a fixed finite voter and candidate set.

\begin{defn}
\label{def:rankings}
We fix an infinite set $C^* = \{c_1, c_2, \dots, \dots \}$ of potential \hl{candidates} and an infinite set $V^*= \{v_1, v_2, \dots, \}$ of potential  
\hl{voters}. Let $C\subseteq C^*$. For each $s\geq 1$, an
\hl{$s$-ranking} is a strict linear order of $s$ elements chosen from
$C$.  The set of all $s$-rankings is denoted $L_s(C)$. When $C$ is finite and $s = |C|$, we
write simply $L(C)$. When $s=1$, we identify $L_1(C)$ with $C$ in the natural way. 
\end{defn}

\begin{remark}
When $C$ is finite, of size $m$ say, the set $L_s(C)$ consists of strict linear orderings of $C$ and has size $m(m-1)\cdots (m-s+1)$. By fixing a default linear 
ordering on $C$, we can interpret elements of $L_s(C)$ as partial permutations of $C$ in the usual way. 
\end{remark}

\begin{defn}
\label{def:profiles}
A \hl{profile} is a function $\pi: V \to L(C)$ where $V\subset V^*$ and $C\subset C^*$ are finite. 
We denote the set of all profiles by $\profs$. An \hl{election} is a triple $(C,V, \pi)$
with $\pi\in \profs$ and $\pi:V \to L(C)$. We denote the set of all elections with fixed $C$ and $V$  by
$\elecs(C, V)$, and the set of all elections by $\elecs$. 
\end{defn}

\begin{remark}
By definition $\pi(v)\in L(C)$ for each $v\in V$. 
If $C$ is linearly ordered as described above, then $\pi(v)^{-1}$ denotes the inverse permutation, and for each 
$c\in C$, $r(\pi(v),c):=\pi(v)^{-1}(c)$ gives the rank of $c$ in $v$'s preference order.

Of course, $C$ and $V$ are implicit in the definition of $\pi$, so strictly speaking an election is completely determined by a profile. 
We distinguish the two concepts because we sometimes want to deal with several different voter or candidate sets at the same time, and because 
$C$ is not really completely determined --- any superset of $C$ would also work.
\end{remark}

\begin{defn}
\label{def:rules}
A \hl{social rule of size $s$} is a function $R$ that takes each election $E = (C,V, \pi)$ to
a nonempty subset of $L_s(C)$. When there is a unique $s$-ranking chosen, the
word ``rule" becomes ``function". When $s=1$, we have the usual
\hl{social choice function}, and when $s=m$ the usual \hl{social welfare
function}.

For each subset $D$ of $\elecs$ we can consider a \hl{partial social rule with domain $D$} 
to be defined as above, but with domain restricted to $D$. We denote the domain of a partial social rule $R$ by $D(R)$. If $R$ and $R'$ are partial social rules such that $D(R) \subseteq D(R')$ and $R(E) = R'(E)$ for all $E \in D(R)$ then we say that $R'$ \hl{extends} $R$.
\end{defn}

\begin{remark}
Most previous work has dealt only with the cases $s=1$ and $s=m$.
\end{remark}

\subsection{Consensus}
\label{ss:consensus}

Intuitively, a consensus is simply a socially agreed unique outcome on some set of
elections. We now define it formally.
\begin{defn}
\label{def:cons}
An $s$-\hl{consensus} is a  partial social function $\cons$ of
size $s$.  The domain $D(\cons)$ of 
$\cons$ is called an \hl{$s$-consensus set}
and is partitioned into the inverse images $\cons_r:=
\cons^{-1}(\{r\})$.
\end{defn}

\begin{remark}
Note that we allow $\cons_r$ to be empty. This happens rarely for natural rules in the distance 
rationalizability framework, because it implies that there is no election for which $r$ is the unique social choice. 
However it is technically useful and allows us to deal with varying sets of candidates.
\end{remark}

It often makes sense to ensure coherence between the various values of $s$ for which we formalize a given consensus notion.

\begin{defn}
\label{def:restrict}
Let $\cons$ be a $1$-consensus. For each $s$ we define an $s$-consensus $\cons_{(s)}$ (the \hl{$s$-restriction} of $\cons$) as follows.  For each candidate $c$, $\cons_c$ is defined. Given $E=(C,V, \pi)\in \elecs$, define $E_{-c}$ to be the election $(C\setminus\{c\}, V, \pi')$, where $\pi'$ is obtained from $\pi$ by erasing $c$ from each ranking. 

Let $D_2$ be the set of all elections $E$ such that both $E=(C, V, \pi)$ and $E_{-c}$ both belong to the domain of $\cons$, where $c = \cons(E)$. Letting $c' = 
\cons(E_{-c})$, define $\cons_{(2)}$  on $D_2$ by its output, the $2$-ranking $cc'$. Continue by induction, reducing the domain at each step if necessary, and output a single $s$-ranking.
\end{defn}

Several specific consensuses have been described in the literature. Here we unify the presentation 
of several of the most common ones. 

\begin{defn}(qualified majority consensus)
\label{def:sunam}

Let $1/2 \leq \alpha < 1$. The \hl{$(\alpha, s)$-majority consensus}
$\unam^{(\alpha, s)}$ is the $s$-consensus with domain consisting of
all elections with the following property: there is some fraction $p> \alpha$ of the voters, 
all of whom agree on the order of the top $s$
candidates. The consensus choice is this common $s$-ranking.

Special cases:
\begin{itemize}
\item When $\alpha =  1/2$, we obtain the usual \hl{majority $s$-consensus} $\maj^s$.
\item The limiting value as $\alpha \to 1$ gives the case of unanimity. We denote this by $\sunam^s$. When $s=|C|$, we simply write $\sunam$ (called the 
\hl{strong unanimity consensus}), 
whereas when $s = 1$, for consistency with previous authors we denote it $\wunam$, the \hl{weak unanimity consensus}. 
\end{itemize}
\end{defn}

\begin{remark}
In general, the $s$-restriction of $\unam^{\alpha, 1}$ is not $\unam^{\alpha, s}$: if a majority rank $a$ first, and a majority of those rank $b$ above all candidates other than $a$, it is not necessarily the case that a majority of votes have $ab$ at the top (the fraction is more than $2\alpha - 1$, however). However,  a majority of the original voters rank $b$ either first or second. The $s$-restriction is the consensus for which more than fraction $\alpha$ of voters agree on the top candidate, more than $\alpha$ agree on the top two, etc.

However, $\sunam^s$ is indeed the $s$-restriction of $\wunam$: if all voters rank $a$ first and all rank $b$ over all  candidates other than $a$, then all agree on the ranking $ab$, etc. 
\end{remark}

\begin{defn}(qualified Condorcet consensus)
\label{def:cond}

Let $1/2 \leq \alpha < 1$. The $\alpha$-\hl{Condorcet consensus} $\cond^{\alpha}$ has
domain consisting of all elections for which an $\alpha$-Condorcet
winner exists. That is, there is a (necessarily unique) candidate $c$ such that for any other candidate $c'$, a fraction strictly greater than $\alpha$ of voters rank $c$ over $c'$.

We define $\cond^{(\alpha, s)}$ to be the $s$-restriction of $\cond^{\alpha}$.

Special cases: 
\begin{itemize}  
\item When $\alpha = 1/2$ we denote this  by $\cond$, the usual Condorcet consensus. 
\item When $\alpha \to 1$, we obtain $\sunam$.
\end{itemize}
\end{defn}

\subsection{Distances}
\label{ss:metrics}

We require a notion of distance on elections. We aim to be as general as
possible. 

\begin{defn}(distance)
\label{def:dist}
A \hl{distance} (or \hl{hemimetric}) on $\elecs$ is a function
$d:\elecs \times \elecs \to \preals \cup \{\infty\}$ that satisfies the
identities 
\begin{itemize}
\item $d(x,x) = 0$,
\item $d(x,z) \leq d(x,y) + d(y,z)$.
\end{itemize}
 A \hl{pseudometric} is a distance that also satisfies 
 \begin{itemize}
\item $d(x,y) = d(y,x)$.
 \end{itemize}
 A \hl{quasimetric} is a distance that also satisfies 
  \begin{itemize}
\item $d(x,y) = 0 \Rightarrow x = y$.
 \end{itemize}
 A \hl{metric} is a distance that is both a quasimetric and a pseudometric.
We call a distance \hl{standard} if $d(E, E') = \infty$ whenever $E$ and $E'$ have
different sets of voters or candidates (this term has not been used in previous literature).
\end{defn}

\begin{eg}
\label{eg:insdel}
Let $d_{del}(E,E')$ (respectively $d_{ins}(E,E')$) be defined as the
minimum number of voters we must delete from (insert into)  election $E$
in order to reach election $E'$ (or $+\infty$ if $E'$ can never be
reached). Each of $d_{ins}$ and $d_{del}$ is a nonstandard  quasimetric.
\end{eg}

\begin{eg} (shortest path distances)
\label{eg:geodesic}
Consider a digraph $G$ with nodes indexed by elements of $\elecs$,
 and some edge relation between elections. Define $d$ to be the (unweighted) shortest path distance in $G$. This is a quasimetric. It is a metric if 
the underlying digraph is a graph. For example, $d_H, d_K, d_{ins}, d_{del}$
are defined via essentially this construction. Note that it suffices to specify for which 
$E, E'$ we have $d(E,E') = 1$ in order to specify such a distance, and not every quasimetric  is a shortest path distance, even after scaling by a constant, because if there are two points at distance $3$ there must also be points at distance $2$, for example. 
\end{eg}

\begin{eg}(some strange distances)
\label{eg:weird dist}
The following distances will be useful for existence results later.
Let $R$ be a rule. 

The first is a metric used by Campbell and Nitzan \cite{CaNi1986}. Define $d$ as follows. 
$$
d(E, E') = 
\begin{cases} 
0 & \text{if $E=E'$}\\
1 & \text{if $|R(E)| = 1$ and $R(E) \subset R(E')$}\\
1 & \text{if $|R(E')| = 1$ and $R(E') \subset R(E)$}\\
2 & \text{otherwise}.\\
\end{cases}
$$

We claim that $d$ is a metric. The only non-obvious axiom is the triangle inequality. It suffices to consider the case where $E, E', E''$ are distinct. Then $d(E,E') + d(E', E'') \geq 1+1=2\geq d(E, E'')$, yielding the result.

The second distance is a variant of the first, where instead we define $d(E, E') = 0$ if and only if $E=E'$ or $R(E) = R(E')$ and $|R(E)| = 1$. This is a pseudometric, since elections with the same unique winner are at distance zero. To prove the triangle inequality, first note that $R(E, E') \leq  1$ if and only if $R(E) \subseteq R(E')$ and $|R(E)| = 1$, or the analogous condition with $E$ and $E'$ exchanged holds. If $2\geq d(E,E'') > d(E,E') + d(E',E'')$, then at least one of the two terms on the right is $0$ and the other is at most $1$. Thus (without loss of generality) $E$ and $E'$ have a common unique winner under $R$ and $R(E') \subseteq R(E'')$, yielding the contradiction $d(E,E'') \leq 1$.

The third distance is the shortest path metric defined as follows: there is an edge joining $E$ and $E'$ 
if and only if $|R(E')| = 1$ and $R(E') \subset R(E)$, or the same with $E$ and $E'$ exchanged (these are the same as the cases defining $d(E,E') = 1$ in the definition of the Campbell-Nitzan distance).
\end{eg}

\subsubsection{Votewise distances}
\label{sss:votewise}

One commonly used class of distances consists of the \hl{votewise}
distances formalized in \cite{EFS2015}, which we now define after some preliminary work. They are each based on
distances on $L(C)$. See \cite{Diac1988} for basic
information about metrics on the symmetric group. 

\begin{eg}
\label{eg:dist Sn}
The most commonly used such distances on $L(C)$ are as follows.
\begin{itemize}
\item the \hl{discrete metric} $d_H$, defined by
$$d_H(\rho, \rho') =  
\begin{cases} 1 \quad \text{if $\rho = \rho'$} \\
0 \quad \text{otherwise}.
\end{cases}
$$
\item the \hl{inversion metric} $d_K$ (also called the swap, bubblesort or Kendall-$\tau$ metric), where $d_K(\rho, \sigma)$ is the 
minimal number of swaps of adjacent elements required to convert $\rho$ to $\sigma$.
\item \hl{Spearman's footrule} $d_S$, defined by 
 $$d_S(\rho, \rho'):= \sum_{c\in C} |r(\rho, c) -r(\rho',c)|.$$
 \end{itemize}
\end{eg} 

\begin{defn}
\label{def:norm}
A \hl{seminorm} on a real vector space $X$ is a real-valued function $N$ satisfying the identities
\begin{itemize}
\item  $N(x+y) \leq N(x)+N(y)$
\item  $N(\lambda x) = |\lambda| N(x)$
\end{itemize}
for all $x,y \in X$ and all $\lambda\in \reals$. Note that this implies that $N(0) = 0$ and 
$N(x) \geq 0$ for all $x\in X$.

A \hl{norm} is  a seminorm that also satisfies 
\begin{itemize}
\item $N(x) = 0 \Rightarrow x=0$.
\end{itemize}
\end{defn}

\begin{remark}
Every seminorm induces  a pseudometric via $d(x,y) = ||x-y||$. This is a metric if and only if the 
seminorm is a norm.
\end{remark}

\begin{eg}
\label{eg:norm}
Consider an $n$-dimensional space $X$ with fixed basis $e_1, \dots, e_n$ and corresponding
coefficients $x_i$ for each element $x\in X$. Fix $p$ with $1\leq p < \infty$ and define the 
$\ell^p$-norm on $X$ by 
$$
||x||_p = \left(\sum_{i=1}^n |x_i|^p\right)^{1/p}.
$$
When $p = \infty$ we define the $\ell^\infty$ norm by 
$$
||x||_\infty =  \max_{1\leq i \leq n} |x_i|.
$$
\end{eg}

\begin{defn} (votewise distances)
\label{def:votewise}

Choose a family $\{N_n\}_{n\geq 1}$ of
seminorms, where $N_n$ is defined on $\reals^n$. Fix candidate set $C$ and voter set $V$, and choose a distance $d$ on $L(C)$. Extend $d$ to a
function on $\profs(C,V)$ by taking $n = |V|$ and defining for $\sigma, \pi
\in \profs(C,V)$ 
$$d^{N_n}(\pi,\sigma):= N_n(d(\pi_1, \sigma_1), \dots, d(\pi_n, \sigma_n)). $$ 
This yields a distance on elections having the
same set of voters and candidates. We complete the definition of the extended distance
(which we denote by $d^N$) on $\elecs$ by declaring it to be standard.

We use the abbreviation $d^p$ for $d^{\ell^p}$, and sometimes we even
use just $d$ for $d^N$ if the meaning is clear.
\end{defn}

\begin{remark}
Note that if $d$ is a metric and $N$ is a norm, then $d^N$ is a metric. 
\end{remark}

\begin{eg} (famous votewise distances)
\label{example:dist}

The  distances $d^1_H$ and $d^1_K$ are 
called respectively the \hl{Hamming metric} and \hl{Kemeny metric}.
The Hamming metric measures the number of voters whose preferences 
must be changed in order to convert one profile to another, and as such has an interpretation in 
terms of bribery. The Kemeny metric measures how many swaps of adjacent candidates are 
required, and is related to models of voter error.
Among the many other votewise metrics, we single out $d^1_S$, sometimes called
the \hl{Litvak distance}. 
\end{eg}

\subsubsection{Tournament distances}
\label{sss:tour dist}

Some distances depend only on the net support for candidates.

\begin{eg} (tournament distances)
\label{example:dist 2}
Given an election $E= (C, V, \pi)$, we form the pairwise majority digraph $\Gamma(E)$ with
nodes indexed by the candidates, where the arc from $a$ to $b$ has
weight equal to the \emph{net support} for $a$ over $b$ in a pairwise contest.
Formally, there is an arc from $a$ to $b$ whose weight equals $n_{ab} - n_{ba}$, where 
$n_{ab}$ denotes the number of rankings in $\pi$ in which $a$ is above $b$.

Let $M(E)$ be the weighted adjacency matrix of $\Gamma(E)$ (with respect to an arbitrarily chosen
fixed ordering of $C$). Given a seminorm $N$ on the space of all $|C| \times |C|$ real matrices, 
we define the $N$-\hl{tournament distance} by 
$$
d^N(E, E') = N(M(E) - M(E')).
$$  

A closely related distance is defined in the
analogous way, but where each element of the adjacency matrix is
replaced by its sign ($1$, $0$, or $-1$).  We call this the
$N$-\hl{reduced tournament distance}. We denote the special cases
where $N$ is the $\ell^1$ norm on matrices by $d^T$ and $d^{RT}$ respectively. A (reduced)
tournament distance cannot be a metric, even if $N$ is a norm, because
it does not distinguish points (the mapping $E \mapsto M(E)$ is not one-to-one). However, it is a pseudometric.
\end{eg}

\subsection{Combining consensus and distance}
\label{ss:combine}

In order for a rule to be definable via the DR construction, it is
necessary that the first following property holds. The second property avoids trivialities and ensures some theorems in Section~\ref{s:foundation} are true. We shall assume both properties
from now on.

\begin{defn} \label{def:distinguish}
Let $d$ be a distance on $\elecs$ and $\cons$ a consensus. Say that $(\cons, d)$
\hl{distinguishes consensus choices} if whenever $x\in \cons_r, y \in \cons_{r'}$
and $r\neq r'$, then $d(x,y) > 0$.
\end{defn}

We use a distance to extend a consensus to a social rule in the natural
way. The choice at a given election $E$ consists of all $s$-rankings $r$
whose consensus set $\cons_r$ minimizes the distance to $E$. We
introduce the idea of a score in order to use our intuition about
positional scoring rules. 

\begin{defn} (DR scores and rules)
\label{def:rules}

Suppose that $\cons$ is an $s$-consensus and
$d$ a distance on $\elecs$. Fix an election $E\in \elecs$. 
The \hl{$(\cons, d, E)$-score} of $r\in L_s(C^*)$ is defined by
$$
|r| : =  \inf_{E'\in \cons_r} d(E, E').
$$
The rule $R:=\R(\cons, d)$ is defined by
\begin{equation}
\label{eq:argmin}
R(E) = \arg\min_{r} |r|.
\end{equation}
We say that $R$ is \hl{distance rationalizable} (DR) with respect to
$(\cons, d)$.
\end{defn}

\begin{remark}
Note that if $\cons_r$ is empty, then $|r| = \infty$.  DR scores are defined so that they are 
nonnegative, and higher score corresponds to larger distance. This is not consistent with
the usual scoring rule interpretation in Example~\ref{eg:scoring}, but the two notions of score 
 are closely related. Our DR scores have the form $M - s$ where $s$ is the score
associated with the scoring rule and $M$ depends on $E$ but not on any $r\in
L_s(C)$.
\end{remark}

\subsection{Some specific rules}
\label{ss:spec rules}

Table~\ref{t:DR egs} presents a few known rules in this framework. Most of the rules in the 
table are well known. We single out the following less obvious references. The \hl{modal ranking rule} was investigated by Caragiannis, Procaccia and Shah \cite{CPS2014}. The \hl{voter replacement rule} 
(VRR) was 
defined essentially as a missing entry in such a table \cite{EFS2012}. 
The  entries marked ``trivial" are so labelled because in those cases every election not in 
$\cons$ is at distance $+\infty$ from every $\cons_r$. Missing entries reflect on the authors' 
knowledge, 
and may have established names. Our table overlaps with that in \cite{MeNu2008} --- note that the 
$(\cond, d_H^1)$ entry is incorrect in that reference, as pointed out by Elkind, Faliszewski
and Slinko \cite{EFS2012}. Our table also overlaps one presented by Elkind, Faliszewski 
and Slinko \cite{EFS2015}.

\begin{table}
\begin{tabular}{c|cccc}

$\cons / d$ & $\sunam$ & $\wunam$ & $\cond$ & $ \cond^m$  \\
\hline
$d_K^1$ & Kemeny & Borda & Dodgson & \\
$d_H^1$ &modal ranking& plurality & VRR & \\
$d_S^1$ &Litvak &Borda&Dodgson &\\
$d_T$ &Kemeny&Borda&maximin& \\
$d_{RT}$ &Copeland &Copeland & Copeland & Slater\\
$d_{ins}$ &trivial &trivial & maximin  & \\ 
$d_{del}$ &modal ranking & plurality &Young & \\
\end{tabular}
\caption{Some known rules in the DR framework (see discussion in Section~\ref{ss:spec rules})}
\label{t:DR egs}

\end{table}

\begin{eg} (scoring rules)
\label{eg:scoring}
The \hl{positional scoring rule} defined by a family of \hl{weight vectors} $w:=w^{(m)}$ satisfying $w_1
\geq \dots \geq w_m, w_1 > w_m$ elects all candidates with maximal score, where
the score of $a$ in the profile $\pi$ is defined as $\sum_{v\in V}
w_{r(\pi(v), a)}$. The positional scoring rule defined by $w$  has the form $\R(\wunam,
d_w^1)$ where $d_w$ is the distance on rankings defined by $$d_w(\rho,
\rho') = \sum_{c\in C} |w_{r(\rho,c)} - w_{r(\rho',c)}|.$$ 
\end{eg}

\begin{remark}
Note that $d_w$ is a metric on $L_s(C)$ if and only if $w_1, \dots, w_s$
are all distinct. The score of $r$ under the rule defined by $w$ is the 
difference $nw_1 - |r|$. For example, for Borda with $m$ candidates (corresponding to $w =(m-1, m-2, \dots, 1, 0)$, the
maximum possible score of a candidate $c$ is $(m-1)n$, achieved only for
those elections in $\wunam_c$. The score of $c$ under Borda is exactly
$(m-1)n - K$ where $K$ is the total number of swaps of adjacent
candidates needed to move $c$ to the top of all preference orders in
$\pi(E)$. 

Plurality (corresponding to $w = (1,0,0, \dots, 0)$) and Borda are special cases, where $d_w$ simplifies to $d_H^1$
and $d_S^1$ respectively. As far as the distance to $\wunam$ or $\cond$ is
concerned, $d_S^1$ and $d_K^1$ are proportional, but 
they are not proportional in general \cite[p. 298--299]{MeNu2008}. 
\end{remark}

\if01
\begin{remark}
Consider the set of all elections for which all
positional scoring rules yield the same unique winner. By convexity of
the set of weight vectors, it suffices to check the finite set of rules
defined by weight vectors $(1, 1, \dots, 1, 0, \dots , 0)$ where the
number of $1$'s is fixed and at most $k-1$; in other words the \hl{$k$-approval rules} for $1\leq
k \leq m-1$. This just describes the Lorenz consensus in another way.
\end{remark}
\fi

\begin{eg} (Copeland's rule)
\hl{Copeland's rule} can be represented as $\R(\cond, d_{RT})$. Indeed,
in an election $E$, the Copeland score of a candidate $c$  (the number
of points it scores in pairwise contests with other candidates) equals
$n-1-s$, where $s$ is the minimum number of pairwise results that must be changed 
for $E$ to change to an election that belongs to $\cond_c$.
\end{eg}


Every rule $\R(\cons, d)$, where $\cons$ is a $1$-consensus, automatically
yields a social rule $\R^s(\cons, d)$ of size $s$ as follows.

\begin{defn}
\label{def:score}
Let $1\leq s \leq m$ and suppose that $\cons$ is a $1$-consensus and $d$
a distance on $\elecs$. We define a social rule $\R^s(\cons, d)$ of size
$s$ by choosing $s$ elements in increasing order of score (if there are ties in the scores, we consider all possible such orderings).
\end{defn}

\begin{remark}
$\R^s(\cons, d)$ is single-valued if and only if the lowest $s$ scores of candidates are distinct.
Note that if the $s$-consensus $\cons'$ is a restriction of the $1$-consensus $\cons$, it is not necessarily the case that $\R(\cons', d)
= \R^s(\cons, d)$.  For example, $\sunam$ is a restriction of $\wunam$, and $\R(\wunam, d_K^1)$ is the social choice rule, Borda's rule. By above, we can also define the social welfare version of Borda's rule. However, $\R(\sunam, d_K^1)$ is Kemeny's rule. The social choice rule obtained by taking the top element of the ranking given by Kemeny's rule is also sometimes called Kemeny's rule. All four rules mentioned here are different. 
\end{remark}

\section{Existence and uniqueness}
\label{s:foundation}

The DR framework is not very restrictive without further assumptions on
$\cons$ and $d$, as shown by Campbell and Nitzan \cite{CaNi1986}. 

\subsection{Existence}
\label{ss:exist}

We give necessary and sufficient conditions, an improvement on
 \cite[Prop. 4.4]{CaNi1986} and \cite[Thm 2]{EFS2015}.
 
\begin{defn}
\label{def:Kmax}
For each rule $R$, there is a unique \hl{maximum consensus}
$\cons^{\max}(R)$, namely that whose consensus set $D^{\max}$ consists of all
elections on which $R$ gives a unique output, which we define as the consensus choice.
\end{defn}

\begin{remark}
Most rules commonly used in practice have ties, so that the domain of $\cons^{\max}(R)$ is 
smaller than the domain of $R$. For example, if a social choice rule satisfies anonymity (symmetry with respect to voters) and neutrality (symmetry with respect to candidates) and is faced with a profile containing exactly one of each possible preference order, it must select all candidates as winners. 
\end{remark}

\begin{defn}
\label{def:nonimp}
The \hl{unique image} of an $s$-rule $R$ is the set of all $r\in L_s(C)$ which occur as the unique 
winner in some election.  That is, there exists $E\in \elecs$ such that $R(E) = \{r\}$. 
The \hl{image} of the rule is the set of all  $r\in L_s(C)$ which occur as a winner in some election. That is, there exists $E\in \elecs$ such that $r\in R(E)$. 
The rule satisfies \hl{nonimposition} if every $r\in L_s(C)$ occurs as a unique winner somewhere --- 
in other words, the unique image of $R$ equals $L_s(C)$.
\end{defn}

\begin{remark}
Although slightly confusing (the image should perhaps be a set of subsets rather than their union) this is the standard terminology for set-valued mappings in mathematics.
\end{remark}

We need to rule out the possibility of an election being at infinite distance from all consensus elections. There is no problem with an election being equidistant from all nonempty consensus sets, but without this assumption, empty consensus sets will (by convention) also be at the same distance.

\begin{defn}
\label{def:nontriv}
Say that $(\cons, d)$ is \hl{nontrivial} if for each $E\in \elecs$ there is some $r$ for which $d(E, \cons_r) < \infty$.
\end{defn}

\begin{prop}
\label{prop:gen}
Let $\cons$ be a consensus and $R$ a rule. There exists a nontrivial distance rationalization 
$R = \R(\cons, d)$ if and only if the following two conditions hold:
\begin{enumerate}[(i)]
\item $R$ extends $\cons$;
\item the image of $R$ equals the image of $\cons$.
\end{enumerate}
Furthermore, $d$ can be chosen to be a metric.
\end{prop}
\begin{proof}
The first condition is necessary: because of the assumption that $\cons$ distinguishes consensus choices, if $E\in \cons_r$ then $d(E, \cons_r) = 0$ but $d(E, F) > 0$ for all $F\in \cons_{r'}$ and all $r'\neq r$. The second condition is necessary: the image of $\cons$ is contained in the image of $R$ because $R$ extends $\cons$, and by the nontriviality assumption, if $\cons_r = \emptyset$ then $r$ is not a winner at any election.

Now assume that the two conditions hold. Let $d$ denote either of the first two distances in 
Example~\ref{eg:weird dist} and let $S = \R(\cons, d)$. We claim that $S = R$ (note that  rationalization is nontrivial because the distances are finite). Let $E\in \elecs$. Since $R$ extends $\cons$ the result is immediate if $E\in \cons_r$ for some $r$, because then $R(E) = \{r\}$ and $S(E) = \{r\}$ since $d$ distinguishes consensus choices.
Now suppose that $E$ is not a member of $\cons_r$ for any $r$. Note that 
 $d(E, \cons_r) \geq 2$ if $r\not\in R(E)$. For each $r\in R(E)$, by assumption $r$ is in the image of $\cons$, so there is $F\in \cons_r$ with $d(E, F) = 1$ (for the first and third distances, or the second if $|R(E)| > 1$) or $F\in \cons_r$ with $d(E,F) = 0$ (for the second distance, if $|R(E)| = 1$). Thus $S(E)$ is precisely the set of $r$ for which $r\in R(E)$, in other words $S(E) = R(E)$.
\end{proof}

\begin{cor}
\label{cor:gen}
Let $R$ be a rule. There exists a distance $d$ and consensus $\cons$ such that $R = \R(\cons, d)$ if and only if the image of $R$ equals its unique image.
\end{cor}
\begin{proof}
Let $\cons = \cons^{\max}(R)$. The image of $\cons$ is precisely the unique image of $R$, and 
$R$ clearly extends $\cons$, so this follows directly from Proposition~\ref{prop:gen}.
\end{proof}

\begin{remark}
In Proposition~\ref{prop:gen}, if we assume that $R$ and $\cons$ satisfy nonimposition, then condition (ii) is satisfied. In this case $R$ is 
distance rationalizable if and only if it extends $\cons$ \cite[Prop. 4.4]{CaNi1986}. In this case, the third distance from Example~\ref{eg:weird dist} can be used. Note that in general the third distance does not work --- consider a rule for which only one consensus set $\cons_a$ is nonempty, yet the rule returns a disjoint two-element set $\{b,c\}$ at some election $E$. The third distance would then yield $\{a,b,c\}$ at $E$, a contradiction.

However the 
assumption of nonimposition is
not necessary --- consider the rule in which $R(E) = \{r\}$ for every
election $E$, and  choose $d$ to be the discrete metric. 
\end{remark}

Thus if $\cons$ is specified, the question of existence is settled. 
For example, every social welfare rule satisfying the usual unanimity axiom (if every voter has the same preference order, the rule outputs precisely this common ranking) can be rationalized with
respect to $\sunam$. 

In view of the flexibility of the DR framework, it  is clear that the key idea is to 
make an appropriate choice of a ``small" $\cons$ and ``natural" $d$ so as to recapture 
rule $R$ via $R= \R(\cons, d)$.

\subsection{Uniqueness}
\label{ss:unique}

We now turn to the question of uniqueness.  The construction in the proof of 
Proposition~\ref{prop:gen} shows that changing both $\cons$ and $d$ can lead to the same rule. When $\cons$ is fixed and $d$ varies, the rule often changes. However it sometimes does not change, as can be seen from Table~\ref{t:DR egs}. A general class of examples where the rule does not change is discussed in Section~\ref{ss:group}.

Similarly, when $d$ is fixed and $\cons$ varies, 
the rule sometimes does not change. For example, consider Copeland's rule, 
which can be described as $\R(\cond, d_{RT})$. It can also be described as 
$\R(\wunam, d_{RT})$, because for each $a\in C$, every point of $\cond_a$ is at 
distance zero from $\wunam_a$ with respect to $d_{RT}$. 
The standard examples such as Borda's, Kemeny's, and Copeland's rules all behave well when we extend the consensus set beyond the one used to define them. This leads us to the following question: if $R$ has the form $\R(\cons, d)$, and $\cons'$ is a consensus that extends $\cons$,  is it necessarily the case that $R = \R(\cons', d)$? In particular, does $R= \R(\cons^\text{max}(R),d)$? The answer is no in general, as we now show.

\begin{eg}
\label{eg:hillas no}
We will define $R = \R(\cons, d)$. First, let $C=\{a,b,c,c'\}$ and let $V$ be a voter set of size $n$. Let $G$ be the graph having the following two connected components. The first one includes all elections where all voters rank $a$ or $b$ first, and all elections where the number of voters ranking $a$ first equals the number of votes ranking $b$ first. The other component contains all other elections. Now, assume that in each component, only elections differing by one vote are linked. Define $d$ to be the shortest path distance defined by $G$.

Now, define the domain of $\cons_{c'}$ to be the second component and the domain of $\cons_c$ to be the elections where everyone ranks $c$ or $c'$ first. Let $\cons_a$ be the set of elections where $a$ or $b$ are ranked first but $a$ gets more first place than $b$ and $K_b$ be the set of elections where $a$ or $b$ are ranked first but $b$ gets more first place than $a$. Let $R=\R(\cons,d)$.

Consider an election $E$ where $a$ and $b$ get an equal number of votes $x/2$. Clearly, if $x=0$, $E$ is in $K_c$. Else if $x=n$, then $E$ is at distance $1$ from $K_c$ and $K_b$. Indeed, $R(E)=\{c\}$ if and only if $x<n/2$ and $R(E)=\{a,b\}$ if and only if $x>n/2$.

So indeed, $\cons^{\max}(R)$ contains all elections except the ones where $a$ and $b$ get an equal number $x\geq n/2$ of votes. Now, let $R'=\R(\cons^{\max}(R),d)$. We still consider elections $E$ where $a$ and $b$ get an equal number of votes $x/2$, but we note that $R'(E)=\{c\}$ if and only if $x<3n/4$. Thus, $R'\not=R$.
\end{eg}

\section{Quotients}
\label{s:quot}

Symmetries of voting rules occur very often in practice. In this section, we
show how to express distance rationalization using only symmetric objects and functions. 
 We start with general equivalence relations, then equivalence relations induced by actions of symmetry groups, and then consider special cases of such actions. In Sections~ \ref{ss:anon},~\ref{s:homog} and ~\ref{ss:neutral}, we apply the general results to anonymity and homogeneity, neutrality and reversal symmetry.

We use a general equivalence relation  $\sim$ on $\elecs$, which we shall specialize in later sections.  All our definitions in this section are understood to be with respect to $\sim$. For example, we 
may refer to ``compatibility" and ``total compatibility" without mentioning $\sim$ directly.

Let $\overline{E}$ denote the equivalence class of $E$, and let $\quot$ denote the 
set of equivalence classes. The usual quotient map $E \mapsto \overline{E}$ takes $\elecs$ onto $\quot$. 

\begin{defn}
\label{def:commute}
Let $R$ be a partial social rule. Then $R$ is \hl{compatible} with $\sim$ if
$R(E) = R(E')$ whenever $\overline{E} = \overline{E'}$. 
\end{defn}

\begin{remark}
In usual mathematical terms, $R$ is compatible with $\sim$ if and only if it is an \hl{invariant} for $\sim$ or a \hl{morphism} for $\sim$.
\end{remark}

\begin{defn}
If $R$ is compatible then we may define a mapping $\overline{R}$ on $\quot$ via $\overline{R}(\overline{E}) = R(E)$ for every
$E\in \elecs$ (it is well-defined precisely because of compatibility of $R$). We call $\overline{R}$  a \hl{partial social rule on $\quot$}. 
\end{defn}

\begin{remark}
We shall apply this construction later, where $\sim$ is the relation defining anonymity or homogeneity, in which case everything makes sense because the projection to the quotient space does not change the candidate sets. However, if the projection does affect the candidate sets (as with neutrality) the result may look strange and the interpretation rather uninteresting, although the theorems will be correct. For example, a rule compatible with the equivalence relation defining neutrality must be the constant rule which chooses the same $r$ at every election, or the rule that chooses all possible $r$ at every election. We discuss this more in Section~\ref{ss:neutral}.
\end{remark}

\subsection{Totally compatible distances}
\label{ss:tot comp}

\begin{defn}
\label{def:tot comp}
A distance $d$  is \hl{totally compatible} with $\sim$ if $d(E, E') = d(F, F')$ whenever 
$\overline{E}= \overline{F}$ and $\overline{E'} = \overline{F'}$.
\end{defn}

\begin{remark}
\label{r:tot comp}
In usual mathematical terms, $d$ is totally compatible if and only if it is an invariant for the equivalence relation $\sim_2:= \left( \sim \times \sim \right)$ on $\elecs \times \elecs$ for which 
$(a,b) \sim_2 (c,d)$ if and only if $a\sim c$ and $b\sim d$. Provided that $\sim$ is not contained in the identity relation (so some equivalence class 
has size greater than $1$), a totally compatible $d$ is not a quasimetric, because whenever $\overline{E} =  \overline{E'}$, necessarily $d(E, E') = 0$. 
\end{remark}

A totally compatible distance relates directly to a distance on $\quot$. The proof of the next result is immediate from the definitions.

\begin{prop}
\label{prop:quot dist}
The set of distances on $\quot$ and the set of totally compatible distances on $\elecs$ are in bijection under the map $d\leftrightarrow \delta$ defined as follows. Given $d$, define $\delta(\overline{E}, \overline{E'}) = d(E, E')$. Given $\delta$, define $d(E, E') = \delta(\overline{E}, \overline{E'})$. 
\noproof
\end{prop}

We want to define DR rules on $\quot$. 

\begin{defn}
\label{def:DR quot}
Let $\delta$ be a distance on $\quot$ and $\cons$ a consensus on $\quot$. The rule $\R(K, d)$ is defined using the analogue of \eqref{eq:argmin}.
\end{defn}

\begin{prop}
\label{prop:DR compat}
The following conditions are equivalent for a social rule $R$ on $\elecs$. 
\begin{enumerate}[(i)]
\item $R$ is compatible and distance rationalizable.
\item $R = \R(\cons, d)$ where $\cons$ is compatible and $d$ is totally compatible.
\item $\overline{R}$ is distance rationalizable on $\quot$.
\end{enumerate}
\end{prop}

\begin{proof} 
Suppose that the first condition holds. 
We use the consensus $\cons:=\cons^{\max}(R)$ from Definition~\ref{def:Kmax} (which is compatible
because $R$ is compatible). We can recapture $R$ as $\R(\cons, d)$ by
defining $d$ to be the second distance in Example~\ref{eg:weird dist}. Let $R' = \R(\cons, d)$. Then if $E \in D(\cons)$, necessarily $R'(E) = R(E)$. If $E \not \in D(\cons)$ then 
$R'(E)$ is precisely the set of $r$ for which $r\in R(E)$, namely $R(E)$. Thus $R' = R$. It remains only to check that $d$ is totally compatible. Since $d$ is defined in terms only of the images $R(E)$ and $R$ is compatible, this follows immediately.

Suppose that the second condition holds. Define $\delta(\overline{E}, \overline{E'}) = d(E, E')$ (this is well-defined since $d$ is totally compatible. Define $K = \overline{K}$. Then $\overline{R} = \R(\overline{K}, \delta)$.

Finally, suppose that the third condition holds. Define $\cons, d$ by composing $K, \delta$ with the projection to $\quot$. Then $R = \R(\cons, d)$ and $R$ is compatible since $R(E) = R(\overline{E})$.
\end{proof}

The distance and consensus used in the proof of Proposition~\ref{prop:DR compat} are rather unnatural (note that the first and third distances in Example~\ref{eg:weird dist} would not even work, being metrics).  We now consider more natural constructions that relate to the original distance and use the equivalence relation explicitly.

\subsection{Quotient distances}
\label{ss:quot dist}

When dealing with equivalence classes, the obvious idea is to use a quotient distance \cite{DeDe2009}. This concept is relatively little-known.

\begin{defn}
\label{def:quot dist}
We define $\overline{d}: \quot \times \quot \to \reals_+$ to be the \hl{quotient distance}
induced by $\sim$. 
\end{defn}

\begin{remark}
The standard construction of quotient distance $\overline{d}$ is as follows:
\begin{equation}
\label{eq:quot}
\overline{d}(x,y)  = \inf \sum_{i=1}^k d(E_i, E'_i)
\end{equation}
where the infimum is taken over all  \emph{admissible paths}, namely
paths such that $E'_i \sim E_{i+1}$ for $1\leq i \leq k-1$, $E=E_1, E' = E'_k$, $E$ projects
to $x$ and $E'$ to $y$. 
\end{remark}

We now focus on a special situation where $\overline{d}$ has a much simpler formula.

\begin{defn}
\label{def:d bar}
Let $d$ be a distance on $\elecs$. Define $\dt$ on $\quot$ by
$$
\dt(x,y) = \inf_{\overline{E} = x, \overline{E'} = y} d(E,E').
$$
\end{defn}
\begin{remark}
If $d$ is totally compatible with $\sim$, then in the definition of
$\dt$, all the distances on the right are the same, so that $\dt(x,y) = d(E, E')$ whenever 
$\overline{E} = x, \overline{E'} = y$. 
\end{remark}

\begin{defn}
\label{def:simple}
A distance on $\elecs$ is \hl{simple} for $\sim$ if $\dt = \overline{d}$.
\end{defn}

\begin{remark}
Every totally compatible distance is simple, because in the definition of 
$\overline{d}$  the minimum is achieved when $k=1$ (use induction on $k$ and the inequality 
$d(E, E_1') + d(E_2, E_2') = d(E, E_1') + d(E_1', E_2') \leq d(E, E_2')$).
\end{remark}

Non-simple distances do arise in our framework.

\begin{eg}
\label{eg:nonsimple}
Let $\sim$ be the equivalence relation on $\elecs$ defined as follows: $E= (C, V, \pi) \sim E'= (C, V, \pi')$ if and only if there are precisely two top-ranked candidates in all the votes in $\pi \cup \pi'$ (each other election forms its own singleton equivalence class). If the two candidates in question are $x, y$, denote by $\elecs_{xy}$ the equivalence class so defined. 

Let $d= d_H$ and suppose $E = (C, V, \pi)\in \elecs_{ab}, E'= (C, V, \pi')\in \elecs_{cd}$, where $\{a,b\} \cap \{c,d\} = \emptyset$. Then $\dt(\overline{E}, \overline{E'})=n$. However 
$\overline{d}(\overline{E}, \overline{E'}) = 2$, since $E$ and $E'$ are each equivalent to elections that are at distance $1$ from $\elecs_{bc}$. Thus $d_H$ is not simple for  $\sim$.
\end{eg}

\begin{prop}\label{prop:simple compat}
Let $d$ be a simple distance and $\cons$ a compatible consensus.  
Then for every $r$ and $E$, $\overline{d}(\overline{E},
\overline{\cons}_r) = d(E, \cons_r)$. 
Thus if $\R(\cons, d)$ is compatible, then it satisfies
$\overline{\R(\cons, d)} = \R(\overline{\cons}, \overline{d})$.
\end{prop}
\begin{proof}
We have 
$$
\overline{d}(\overline{E}, \overline{\cons}_r)  
 = \min_{E'\in \cons_r} \overline{d}(\overline{E}, \overline{E'}) 
= \min_{E'\in \cons_r} \dt(\overline{E}, \overline{\cons}_r)
 = \min_{E'\in\cons_r} \min_{E''\sim E'} d(E, E'') 
 = \min_{E''\in \cons_r} d(E, E'') = d(E, \cons_r).
$$
The first equality holds by definition of distance to a set, the second because $d$ is simple, the 
third by definition of $\dt$, the fourth by compatibility of $\cons$ and the fifth for the same reason as 
the first.

Now let $S = \R(\overline{K}, \overline{d})$. Then $S(\overline{E}) = R(E)$ by above, for all $E$. If $R:=\R(\cons, d)$ is compatible then $\overline{R}$ exists and $\overline{R}(\overline{E}) = R(E)$ for all $E$. Thus $S=\overline{R}$.
\end{proof}

\begin{remark}
The condition that the distance be simple is necessary. For example, consider the setup of Example~\ref{eg:nonsimple}, where the consensus sets have the form $\elecs_{xy}$. 
\end{remark}

\begin{remark}
The condition that $\R(\cons, d)$ is compatible is not always satisfied, as we see when studying homogeneity in Section~\ref{s:homog}. In Section~\ref{ss:group} we give sufficient conditions for it to be satisfied automatically. 
\end{remark}

\subsection{Symmetry groups}
\label{ss:group}

Equivalence is a form of symmetry between elections. Proposition~\ref{prop:simple compat} is clearly useful, but the only simple distances we have seen so far are 
totally anonymous ones, for which the result is obvious. We introduce a strengthening of equivalence that will yield simple distances. We apply it 
in later subsections to discuss anonymity, neutrality, reversal symmetry and homogeneity.

We recall some basics of the theory of group actions on sets. Let $X$ be a set and $G$ a subgroup of the group of all permutations of $X$. The \hl{orbit} of $x\in X$ under $G$ is the set of all $g(x)$ as $g$ ranges over $G$.

\begin{defn}
\label{def:equiv}
Let $\sim$ be an equivalence relation on $\elecs$ and let $G$ be a group acting on $\elecs$ via morphisms. In other words, for each $g\in G$, $E \sim E'$ implies $g(E) \sim g(E')$.
If the equivalence classes of $\sim$ are precisely the orbits under the action of $G$, then we say $\sim$ is \hl{induced by} $G$.

The distance $d$ is \hl{$G$-equivariant} if $G$ acts via isometries:
$$
d(g(E), g(E')) = d(E, E') \qquad \text{for all $E, E' \in \elecs, g\in G$}.
$$
The partial social rule $\cons$ is \hl{$G$-invariant} if $R(g(E)) = R(E)$ for all $E\in D(\cons), g\in G$.
\end{defn}

\begin{prop}
\label{prop:DR invariant}
Suppose that $G$ is a group that induces $\sim$ via an action on $\elecs$. The following conditions are equivalent for a social rule $R$.
\begin{enumerate}[(i)]
\item $R$ is $G$-invariant and distance rationalizable.
\item $R = \R(\cons, d)$  where $\cons$ is $G$-invariant and $d$ is $G$-invariant.
\item $R = \R(\cons, d)$ where $\cons$ is $G$-invariant and $d$ is $G$-equivariant.
\item $\overline{R}$ is distance rationalizable on $\quot$.
\end{enumerate}
\end{prop}
\begin{proof}
The first, second and fourth parts are equivalent by Proposition~\ref{prop:DR compat}. The second implies the third by definition. Suppose that the third condition holds. It remains to show that $R$ is $G$-invariant. Fix arbitrary $r$ and $g\in G$. Since $\cons$ is $G$-invariant, $g(\cons_r) = \cons_r$.
Then $d(E,\cons_r)= d(g(E), g(\cons_r)) = d(g(E), \cons_r)$. Thus $R(E) = R(g(E))$, yielding the second condition.
\end{proof}

However, the proof of Proposition~\ref{prop:DR invariant} does not give a relationship between the distances used in parts (ii) and (iii). We now proceed to clarify this. We first give an important sufficient condition for a distance to be simple.

\begin{prop}
\label{prop:simple}
Let $G$ be a group, and suppose that $\sim$ is induced by $G$, while $d$ is $G$-equivariant. Then $d$ is simple. 
Furthermore
$$
\overline{d}(\overline{E}, \overline{E}') = \min_{g\in G} d(E,g(E')) = \min_{g\in G} d(g(E), E').
$$
\end{prop}
\begin{proof} 
We show that for each $x,y\in \quot$, the minimum value of $k$ for
paths achieving the minimum in \eqref{eq:quot} is always $1$. Assume that this is not the case, so
there exist $x,y$, a minimum $k>1$ and admissible paths such that
$$
\overline{d}(x,y) = \sum_{i=1}^k d(E_i, E'_i).
$$

Choose  $g,h\in G$ so that $g(E_k) = E'_{k-1}$ (possible since $\sim$ is induced by $G$).
Then by  $G$-equivariance and the triangle inequality
\begin{align*}
d(E_{k-1}, E'_{k-1}) + d(E_k, E'_{k}) &= d(E_{k-1}, E'_{k-1}) + d(g(E_k), g(E'_{k})) \\ & =  d(E_{k-1}, E'_{k-1}) + d(E'_{k-1}, g(E'_k)) \\
&\geq d(E_{k-1}, g(E'_k)).
\end{align*}
This contradicts the minimality of $k$, and this contradiction shows that $d$ is simple. 
The other formulae for $\dt$ 
follow immediately, because  all $E'$ projecting to
$y$ are equivalent, so that each can map onto any other via some $g$.
\end{proof}

\begin{prop}
\label{prop:nice dist}
Suppose that $R = \R(\cons, d)$ where $\cons$ is $G$-invariant and $d$ is $G$-equivariant.
Then $\overline{R} = \R(\overline{\cons}, \dt)$.
\end{prop}
\begin{proof}
$R$ is $G$-invariant by Proposition~\ref{prop:DR invariant}, so $\overline{R}$ exists. By Proposition~\ref{prop:simple}, $d$ is simple. The result follows from Proposition~\ref{prop:simple compat}.
\end{proof}

\begin{remark} Also $R = \R(\cons, d')$ 
where $d'$ is $G$-invariant and $\overline{d'} = \dt$.  In fact 
$d'(E,E')$ equals  the familiar quantity $\min_{g,g'\in G} d(g(E), g'(E'))$.
\end{remark}

\subsection{Neutrality and reversal symmetry}
\label{ss:neutral}

Proposition~\ref{prop:DR invariant} does not apply to the study of these properties, because rather than $R(g(E)) = R(E)$, we want $R(g(E)) = g(R(E))$: the output of the rule changes in a consistent way. This is because the group action changes the candidates, unlike the case with anonymity.

A social rule of size $s$ is a mapping taking each election to a set of $s$-rankings. If a group $G$ acts on the set of rankings, then there is a natural induced action on social rules: $g(R)$ is the rule for which $g(R)(E) = g(R(E))$. If the group acts on $s$-rankings for all $s$, we can say more.

\begin{defn}
Suppose that $G$ is a group acting on $L(C^*)$ such that it maps $L_s(C^*)$ to itself for every $s$. The partial social rule $R$ is \hl{$G$-equivariant} if the identity $R(g(E)) = g(R(E))$ holds.
\end{defn}

An example of this is \hl{reversal symmetry}. A social welfare rule satisfies reversal symmetry if turning all input rankings upside down also reverses the output. The group in question is the group of order $2$, generated by $g$ say. All our examples so far of distances on rankings have satisfied reversal symmetry. 

\begin{remark}
Reversal symmetry has been defined for social choice rules as: if $a$ is the unique winner in the original profile, then $a$ is not a winner in the reversed profile.
For example, the social welfare version of Borda's rule satisfies this. However for social choice rules which do not come from social welfare rules, this is inconsistent with our definition above. In the case of DR rules, there is no difficulty, because each social choice rule yields a social welfare rule as in Definition~\ref{def:score}.
\end{remark}

It is straightforward to prove an analogue of Proposition~\ref{prop:DR invariant} that extends to the case of $G$-equivariant rules. We omit the details.

\begin{prop}
\label{prop:reversal}
If $\cons$ and $d$ satisfy reversal symmetry then so does $\R(\cons, d)$. 
Conversely, if $R$ satisfies reversal symmetry then $R = \R(\cons, d)$ where $\cons$ and $d$ both do.
\noproof
\end{prop}

For example, since $\sunam, d_H^1$ and $d_K^1$ each satisfy reversal symmetry, so do the modal ranking rule and Kemeny's rule. 

If a group $G$ acts on $C^*$ then there is a natural induced action on $s$-rankings, whereby the ranking $a_1\cdots a_s$ maps to $g(a_1)\cdots g(a_s)$. An example of this involves \hl{neutrality}. In this case $G$ the group of all permutations of the candidates. 
Neutrality is a very natural
condition for consensuses and for distances, and is satisfied by all our main examples. It means that the identities of candidates are not relevant because each candidate is treated symmetrically. 

\begin{prop}
\label{prop:neut}
Let $\cons$ be a neutral consensus and $d$ a neutral distance. Then $\R(\cons, d)$ is neutral. Conversely, if $R$ is neutral and distance rationalizable then $R = \R(\cons, d)$ where 
$\cons$ and $d$ are neutral.
\noproof
\end{prop}

\subsection{Anonymity}
\label{ss:anon}

We now apply Proposition~\ref{prop:DR invariant} directly.
First we discuss the concept of anonymity. Several authors use a fixed finite 
voter set 
and define a rule to be anonymous if the rule is invariant under permutations 
of the set. This deals with the order of voters, but not their identities. 
On the other hand, allowing arbitrary identities leads us to issues of 
classes that are not sets, category theory, etc. Our convention that there
is a single countably infinite set of voters allows us to deal both with 
the order and identity of voters.

We start with an example that explains the need to distinguish 
$G$-equivariance from total compatibility.

\begin{eg} \label{eg:anon dist}
Consider $R:=\R(\wunam, d^1_H)$, plurality rule. The consensus $\wunam$
is anonymous, because identities and order of voters are not important. 
The distance $d^1_H$ is \emph{not}
totally anonymous, but it is anonymous. Consider the case $E = (C,V,\pi), E' = (C,V, \pi')$, 
where $C = \{a,b\}$, $V = \{v_1, v_2\}$,
$\pi = (ab, ba)$, $\pi'  = (ba, ab)$. Then $d(E, E') \neq 0 = d(E,E)$, although each of $E$ and $E'$ is obtained from the other by permutation of the
set of voters. 
\end{eg}

We can use the results of Section~\ref{ss:group} directly.

\begin{defn}
\label{def:anon2}
Let $G$ be the group of all bijections of $V^*$. For any set $X$, we define an action of $G$ on functions in the usual way by $g\cdot f (v) = f(g(v))$. 
In particular for each $C$ we can apply this to $X = L(C)$. This allows us to define an action on $\elecs$ via  $g\cdot (C, V, \pi):= (C, g(V), g\cdot \pi)$. Let $\sim$ be the 
equivalence relation induced by this action. A partial rule is  \hl{anonymous} if it is compatible with $\sim$. A distance is \hl{anonymous} if 
 it is $G$-equivariant, and  \hl{totally anonymous} if it is totally compatible with $\sim$.

We denote $\quot$ by $\votesits$ and call it the set of \hl{anonymous profiles} or \hl{voting situations}. 
\end{defn}

\begin{remark}
An anonymous profile is completely determined by the numbers of voters having each preference order, and hence is encoded by 
a \hl{multiset} $\nummap(E)$ on $L(C)$ of weight $|V|$ (we call $\nummap$ the \hl{vote number map} --- note that $E \sim E'$ if and only if $\nummap(E) = \nummap(E')$). 
A rule is anonymous if its output depends only on the anonymous profile, and not the particular voters or their order.
\end{remark}

\begin{eg}
\label{eg:anon}
Let $C=\{c_1, c_2\}$. Let $g$ be the bijection of $V^*$ that transposes $1$ and $2$. Let $E = (C, \{v_1\}, \{ba\}), E' =  (C, \{v_1, v_2\}, \{ba,ab\})$. Then $g(E)$ is an election in which $v_1$ votes $ab$ and $v_2$ votes $ba$, while in $g(E')$, $v_2$ votes $ab$ and $v_1$ does not vote.
\end{eg}

\subsubsection{(Totally) anonymous distances}
\label{sss:anon dist}

The next result is obvious, but useful. Part (ii) was observed by Elkind, Faliszewski and Slinko \cite[Proposition~1]{EFS2015}.

\begin{prop}
\label{prop:tot anon}
The following results hold.
\begin{enumerate}[(i)]
\item Every (reduced) tournament distance is totally anonymous.
\item A votewise distance is anonymous if and only if its underlying seminorm is symmetric.
\item A votewise distance based on a norm 
cannot be totally anonymous.
\end{enumerate}
\end{prop}
\begin{proof}
The first part is clear because the definition uses only the number of 
votes of each type. The second follows immediately from the definitions. For the third part, use the idea of Example~\ref{eg:anon dist}.
\end{proof}

\begin{eg}
\label{eg:votewise anon}
For an anonymous votewise distance, we can let $S=\{s_1^{n_1}, \dots, s_k^{n_k}\}$
denote the multiset of weight $n$ corresponding to the $n$-tuple $(a_1,
\dots, a_n)\in \reals^n$. We can define $N(S) = N_n(a)$, where $n =
\sum_i n_i s_i$ can be computed knowing only $S$.

For example, consider the $\ell^p$ norm for $1\leq p \leq \infty$ defined on $\reals^n$. This yields an anonymous 
votewise distance when coupled with any underlying distance $d$ on $L(C)$.
\end{eg}

\begin{eg} \label{eg:Hamming-quot}
Consider the Hamming distance $d:=d_H^1$. For each $x,y\in \votesits$,
the distance $\overline{d}(x, y)$ is the minimum number of voters whose
votes must be changed in order to transform $x$ into $y$. For example if
$x\in \votesits$ has $2$ $abc$ voters and $3$ $bac$ voters,
while $y$ has $2$ $bac$ voters and $3$ $cba$ voters, then
$\overline{d}(x,y) = 3$. Note that for
the Kemeny metric $d:=d_K^1$, $\overline{d}(x,y) = 8$. 
\end{eg}

\subsubsection{Anonymous DR rules}
\label{sss:anon rule}

Because of the special form of the equivalence relation for anonymity (it does not touch the candidate sets), a partial social rule on  $\votesits$ has a nice form. Indeed, many authors define voting rules directly on $\votesits$.
We simply define an anonymous rule in the DR framework by choosing a consensus notion $K$ on $\votesits$ and a distance $\delta$ on $\votesits$, and using the analogue of~\eqref{eq:argmin}.
Because $\votesits$ can be described as multisets which correspond geometrically to histograms, this may allow us to create interesting anonymous rules using geometric intuition.

The next characterization, which follows directly from Propositions~\ref{prop:DR invariant} and~\ref{prop:nice dist}, answers positively a question raised in \cite[p. 362, discussion after Prop. 4]{EFS2015}. 

\begin{prop}
\label{prop:anon}
If $\cons$ and $d$ are anonymous, then $\R(\cons, d)$ is anonymous and $\overline{R} = \R(\cons, \dt)$. Conversely if $R$ is anonymous and distance rationalizable, then $R = \R(\cons', d')$ where $\cons$ is anonymous and $d'$ is totally anonymous.
\noproof
\end{prop}

This applies to all consensuses described so far, and to all votewise distances based on symmetric seminorms, in addition to tournament distances. Thus all rules in Table~\ref{t:DR egs} are anonymous.

Note that elements of $\votesits$ can be encoded by multisets which are essentially histograms. A standard measure of distance between histograms is the \hl{Earth Mover} or \emph{transportation} distance. The interpretation in our situation, when $d$ is anonymous 
and standard, is that we must move voter mass between types of voters while incurring 
the minimum cost (distance). In fact in this case $\dt$ is exactly the Earth Mover distance based on $d$. Computing it is a special case of the \emph{linear assignment problem} of operations research. The minimum can be computed in polynomial time  via the ``Hungarian method" \cite{Munk1957}. An equivalent formulation of the problem is to find a minimum weight matching in a bipartite graph.

\section{Homogeneity}
\label{s:homog}

In this case we use a slightly different equivalence relation.

\begin{defn}
\label{def:distmap}
Let $E = (C,V, \pi)$ be an election, where $n=|V|$. 
The \hl{vote distribution} associated to $E$ is the probability distribution on $L(C)$ induced by 
the multiset $\nummap(E)$, which we denote $\distmap(E)$. The vote distribution map defines an equivalence relation $\sim$ on $\elecs$ in the usual way. We denote the 
quotient space by $\votedists$. 
\end{defn}

\begin{defn}
\label{def:homog}
A rule is \hl{homogeneous} if and only if it is compatible with $\sim$. A  distance is called \hl{totally homogeneous} if it is totally compatible with $\sim$.
\end{defn}

\begin{remark}
In other words, for a homogeneous rule, the set of winners depends only on the probability 
distribution of voter types --- cloning each voter the same number of times makes no difference to the result.  
Our definition of homogeneity implies anonymity, because the equivalence relation used in 
this section refines the one used for anonymity. Some authors do not make it clear whether they consider homogeneous rules to be anonymous, because they give a definition in terms of profiles in which the cloned voters occupy a particular position. Of course, the two definitions are the same in the presence of anonymity.

It is important to note that $\sim$ is not induced by a group action. Rather, there is a monoid (a ``group without inverses") acting. On $\votesits$ there is an action of the positive integers under multiplication 
where for each $x\in \votesits$, $k\cdot x$ is the voting situation formed by adding $k-1$ copies of each voter.  A rule is homogeneous if it is anonymous and invariant under the action of this monoid. Thus, for example, starting with an election $E$ and doubling or tripling the number of voters will lead to equivalent elections $2E, 3E$, but there is not 
necessarily any way to get from  $3E$ to $2E$ via an element of the monoid, because of the lack of inverses. This has important consequences, as we now see.
\end{remark}

The above remark shows that Proposition~\ref{prop:DR invariant} does not necessarily apply to homogeneity. In fact the conclusion is known to be false.

\begin{eg}
\label{eg:dodgson}
Consider $\cons = \cond$ and $d = d^1_K$. The rule $\R(\cons, d)$ is known as \hl{Dodgson's rule} 
and is known not to be homogeneous, although it is anonymous. For example, consider the following example of Fishburn \cite{Fish1977} with $C=\{a_1, \dots, a_7, x\}$. We start with $a_1\dots a_7$, and consider all its $7$ cyclic permutations. We then insert $x$ between the $4$th and $5$th entries in each case, so $x$ is always in $5$th position.
Then $d^1_K(E, \cond_x) = 7$ ($x$ must switch past each $a_i$ exactly once) but $d^1_K(E, \cond_{a_i}) = 6$ for each $i$, because, for example, $x$ must switch past $a_7, a_6, a_5$ respectively $3,2,1$ times.

However, let $k\geq 1$ and consider the  election $kE$. Then $k^{-1}d^1_K(kE, \cond_x) \to 3.5$ as $k\to \infty$ (because we need only just over $1/2$ a switch per $a_i$), while $k^{-1}d^1_K(kE, \cond_{a_i}) \to 4.5$ (because, for example, $a_1$ must switch past $a_7, a_6, a_5$ respectively just over  $2.5,1.5,0.5$ times).

In the analogue of the proof of Proposition~\ref{prop:DR invariant}, we can conclude only that $kE$
minimizes the distance to elements of $D(\cons)$ of the form $kE'$, but not to all of $D(\cons)$.
\end{eg}

\begin{remark}
The same example shows that $\R(\cond, d^1_H)$, the Voter Replacement Rule, is not homogeneous. In this case the limiting distances to $\cond_x$ and $\cond_{a_i}$ are $1.75$ and $1.5$, while the distances for the original $E$ are both equal to $2$. 
\end{remark}

In order to prove a result similar to Proposition~\ref{prop:DR invariant}, we need a strong condition on $\cons$.

We call an anonymous consensus \hl{divisible} if every element of $\cons_r$ with $kn$ voters has the 
form $kE$ where $E$ has $n$ voters. This is a very strong condition --- taking $n=1$ shows that $\cons$
 is extended by $\sunam$ (up to possible permutation of the winners).

We now generalize \cite[Thm 8]{EFS2015}, which dealt with the case where 
$\cons = \sunam$ and $d$ is  $\ell^p$-votewise, based on an underlying pseudometric.

\begin{defn}
\label{def:homog dist}
An anonymous  distance on $\elecs$ is \hl{homogeneous} if for each $k\geq
1$ and each $E,E'\in \elecs$, 
$$d(E, E') = d(kE, kE').$$ 
A family of symmetric seminorms $N$  is \hl{homogeneous} if $N_{nk}(x^{(k)})
= N_n(x)$ for all $x\in \reals^n$ and all $k\geq 1$. Here $x^{(k)}$
denotes the element of $\reals^{nk}$ obtained by concatenating $k$
copies of $x$.
\end{defn}

\begin{remark}
The reader should avoid confusion by noting 
that the term \emph{homogeneous} is often used for the different property of a seminorm 
expressed by the identity $N(\lambda x) = |\lambda| N(x)$.
\end{remark}

\begin{remark}
Let $d$ be a standard distance and $N$ a symmetric seminorm. Then $d^N$ can be normalized to be homogeneous. Explicitly, let $d_*^N(E, E') = n^{-1} d^N(E, E')$ where $E = (C,V, \pi)$ and $|V| = n$. The DR rules defined by $d^N$ and $d_*^N$ are the same, since we are only scaling the distance by a constant factor.
\end{remark}

\begin{prop}
\label{prop:divis}
Let $\cons$ be a homogeneous divisible consensus and $d$ a homogeneous distance. Then 
$\R(\cons, d)$ is homogeneous.
\end{prop}
\begin{proof}
The proof of Proposition~\ref{prop:DR invariant} adapts directly to this case, as described above.
\end{proof}

Thus we recapture the well-known fact that Kemeny's rule $\R(\sunam, d_K^1)$ is homogeneous. Proposition~\ref{prop:divis}  shows, for example, that 
although Dodgson's
rule can be rationalized with respect to $\sunam$ and some distance (since Dodgson's rule satisfies the unanimity axiom), 
no such distance can be homogeneous.

\section{The Votewise Minimizer Property}
\label{s:VMP}

The failure of $\R(\cons, d)$ to inherit various conditions from $\cons$ is related to the fact that minimization does not respect various operations. Roughly speaking, votewise distances combine better with votewise consensuses. We now make some technical (and rather strong) definitions that allow for several positive results when dealing with votewise distances.

\begin{defn}
\label{def:CMP}
Let $\cons$ be a compatible consensus and $d$ a compatible distance. Say that $(\cons, d)$ has
the \hl{compatible minimizer property} (CMP) if for each $E, E'\in \elecs$ with $E\sim E'$ and each $r$, 
$d(E, \cons_r) = d(E', \cons_r)$.

\begin{remark}
If $\sim$ is induced by a group action then the CMP is automatically satisfied, as used in the proof of Proposition~\ref{prop:DR invariant}. 
The analogue of Proposition~\ref{prop:DR invariant} does not hold for general equivalence relations, as we see in Example~\ref{eg:dodgson}. However, with the additional assumption of the CMP, everything works well. 
\end{remark}

\begin{prop}
\label{prop:CMP}
Let $\cons$ be a compatible consensus and $d$ a compatible distance, and suppose that $(\cons, d)$ satisfies the CMP. Then $\R(\cons, d)$ is compatible.
\end{prop}

\begin{proof}
Let $E, E'\in \elecs$ with $E\sim E'$. By CMP, $d(E, \cons_r) = d(E', \cons_r)$ for all $r$ and in particular the minimizing values of $r$ are the same. 
\end{proof}
\end{defn}

\begin{eg}
\label{eg:cond homog}
If $d$ is totally compatible then the CMP is automatically satisfied. Thus, for example, every rule $\R(\cons, d)$, where $d$ is a tournament distance and $\cons$ is anonymous and homogeneous, is anonymous and homogeneous.
\end{eg}

\begin{eg}
\label{eg:VMP no}
Consider the
election $E=(C,V,\pi)$ where $C=\{a,b\}$, $V$ has size $5$, and $\pi =
\{ab,ab,ba,ba,ba\}$. Then $d(E, \cond_a) = 1$ for $d\in \{d_H, d_K\}$,
and every minimizer differs from $\pi$ only in that precisely one of the
$ba$ voters switches to $ab$. However, if we consider $3E$ then each minimizer requires not $3$, but $2$ switches. Thus $(\cond, d)$ 
does not satisfy the CMP with respect to the equivalence relation used to define homogeneity.
This also shows that
$(\maj,d)$  need not satisfy the CMP, because $\maj$
coincides with $\cond$ when $m=2$.

Thus we should not necessarily expect Dodgson's rule or the Voter Replacement Rule to be homogeneous, and indeed they are not, as Example~\ref{eg:dodgson} shows.
\end{eg}

\begin{defn}
\label{definition:vote min} 
Suppose that $d$ is votewise and anonymous, and $\cons$ is anonymous.
Say that $(\cons, d)$ satisfies the \hl{votewise minimizer property} (VMP) if the following condition is satisfied.
\begin{quotation}
For each  $r\in L_s(C)$ and each election $E = (C,V,\pi)\in \elecs$, there
exists a minimizer $(C,V,\pi^*)\in \cons_r$ of the distance from $E$ to
$\cons_r$, so that for all $i, d(\pi_i,\pi^*_i)$ depends only on $\pi_i$
and $r$.
\end{quotation}
\end{defn} 

\begin{prop}
\label{prop:VMP char}
If $(\cons, d)$ satisfies the VMP, then 
\begin{itemize}
\item $d(\pi_i, \pi_i^*)$ has the form $\delta(t,\rho)$ for some function $\delta$, where $t,\rho \in L(C)$;
\item $d(E, \cons_r)$ has the form $N(S)$ where $S$ is the
multiset of all values of $\delta(t, \rho)$ counted with multiplicity.
\end{itemize}
\end{prop}

\begin{proof}
This follows directly from the definitions. 
\end{proof}

\begin{eg}
\label{eg:VMP formula}
Let $\cons = \wunam$ and $d = d_K$, and $N = \ell^2$. For each $E=(C,V,\pi)\in \elecs$ and $a\in C$, $d(E, \wunam_a) = N(d(\pi_1, \pi_1^*), \dots , d(\pi_n, \pi_n^*))$.
We can take $\pi^*$ to be the ranking derived from $\pi$ by swapping $a$ to the top. Thus 
$d(E, \wunam_a)^2 = \sum_{t\in L(C)} n(t) r(t,a)^2$, where $n(t)$ is the number of times $t$ occurs in $\pi$.

\end{eg}

\begin{prop}
\label{prop:VMPimpCMP}
Let $\cons$ be an anonymous consensus and $d$ a votewise anonymous and homogeneous distance.
If $(\cons, d)$ satisfies the VMP, then it satisfies the CMP with respect to the equivalence relation defining homogeneity.
\end{prop}

\begin{proof}
For each $E$ and $r$, there is a minimizer of $d(E, \cons_r)$ for which the distance has the form $N(S)$ where $S$ is the multiset of values of $d(\pi_i, \pi'_i)$ occurring. Thus it depends only on the equivalence class with respect to anonymity. By homogeneity of $d$, it in fact only depends on the equivalence class of with respect to homogeneity.
\end{proof}

Example~\ref{eg:VMP no} shows that the VMP is not always satisfied, and Proposition~\ref{prop:VMP} gives sufficient conditions for it to be satisfied.

\begin{prop}
\label{prop:VMP}
Let $d$ be an anonymous votewise distance on $\elecs$.
Suppose that the $s$-consensus $\cons$ satisfies the following: for each $r\in L_s(C)$, there is a nonempty subset $S_r$ of $L(C)$ such that
$\cons_r$ consists precisely of elections for which no voter has a ranking in $S_r$.
Then $(\cons, d)$ satisfies the VMP.
\end{prop}

\begin{proof}
The minimizer in question is obtained by, for each $i$,
choosing the closest element of $L(C)$ under the underlying distance. 
\end{proof}

\begin{eg}
\label{eg:VMP yes}
$(\sunam^s, d)$ satisfies the VMP for each $s$, because we can take $S_r$ to be the set of rankings which do not agree with $r$ in all of their top $s$ places. Any consensus which $\sunam^s$ extends also satisfies VMP. For example, we can choose one fixed ranking that does agree with $r$ in the top $s$ places, and define $S_r$ to be its complement. Note that this example is not neutral.

\end{eg}

\subsection{Homogeneity}
\label{ss:VMP homog}

So far we can only show homogeneity when using $\sunam$. We want to widen this to at least $\wunam$. We use a definition from Elkind, Faliszewski and Slinko \cite{EFS2015}.

\begin{defn}
We call a seminorm $N$ \hl{monotone in the positive orthant} if whenever 
$0\leq x_i \leq y_i$ for all $i$, $N(x) \leq N(y)$. 
\end{defn}

\begin{prop}
\label{prop:homog VMP} Suppose that $\cons$ is homogeneous, $d^N$ is votewise, anonymous and homogeneous, $(\cons, d^N)$ satisfies the VMP, and $N$ is monotone in the positive orthant. 
Then $\R(\cons, d^N)$ is homogeneous. 
\end{prop}

\begin{proof}
Let $E\in \votesits$ and $k\geq 1$. First note that $d(E, \cons_r) = d(kE, k\cons_r) \geq d(kE, \cons_r)$. We now prove the converse inequality. 

By VMP, $d(E, \cons_r) = N(S)$ where $S$ is the multiset of values $d(\pi, \pi^*)$. Also by VMP and homogeneity, $d(kE, \cons_{r}) = N(kS')$ where $S'$ is the multiset of values $d(\pi, \pi^{**})$ (here the minimizer may depend on $k$, so $\pi^{**}$ may not equal $\pi^*$). Note that $S'$ is elementwise at least as great as $S$, because $\pi^*$ is a minimizer. By homogeneity and monotonicity in the positive orthant, $d(kE, \cons_{r}) = N(S') \geq N(S) = d(E, \cons_r)$, as required.

Proposition~\ref{prop:CMP} now gives the result.
\end{proof}

\begin{cor}
\label{cor:homog}
If $1\leq p \leq \infty$, then
$\R(\sunam^s, d^p)$ is homogeneous. 
\end{cor}

\begin{remark}
The case $p = \infty$ allows for  stronger results \cite[Thm 9]{EFS2015}.
\end{remark}

\subsection{Consistency}
\label{ss:consist}

Consistency, introduced by Young \cite{Youn1975}, deals with the effect of splitting the voter set into two parts.

\begin{defn}
\label{def:consist}
Let $E = (C, V, \pi)$ and $E' = (C, V', \pi') \in \elecs$ where $V \cap V' = \emptyset$. We define $E+E' = (C, V \cup V', \pi'')$, where 
$$
\pi''(v) = \begin{cases} \pi(v) \qquad \text{if $v \in V$;} \\ \pi'(v) \qquad \text{if $v \in V'$.} \end{cases}
$$
A partial social rule $R$ is \hl{consistent} if whenever 
$R(E) \cap R(E') \neq \emptyset$, necessarily $R(E) \cap R(E') = R(E+E')$.
\end{defn}

\begin{remark}
 A consensus is consistent if and only if each consensus set is closed under the $+$ operation.
\end{remark}

The next result generalizes Elkind, Faliszewski and Slinko \cite[Thm 7]{EFS2015}.
\begin{prop}
\label{prop:consist VMP}
Suppose that $\cons$ is consistent, $d$ is votewise with respect to a homogeneous norm and 
$(\cons,d)$ satisfies the VMP. Then $\R(\cons, d)$ is consistent.
\end{prop}
\begin{proof}
Let $E, E'\in \elecs$ such that $R(E) \cap R(E') \neq \emptyset$. We show that for all $r$ there are minimizers $m(E, r), m(E', r)$ and $m(E+E', r)$ such that $m(E,r) + m(E',r) = m(E+E', r)$. The result then follows just as in Proposition~\ref{prop:CMP}.

The claim follows easily from the VMP. Because minimization of the distance to $r$ occurs votewise, it 
respects the split into $E$ and $E'$.
\end{proof}

\begin{cor}
\label{cor:consist}
If $1\leq p \leq \infty$, then
$\R(\sunam^s, d^p)$ is consistent. 
\end{cor}

Recall that Kemeny's rule is consistent when properly considered as a social welfare rule, but not when considered as a social choice rule  (this point may potentially confuse readers of \cite{EFS2015}).

\subsection{Continuity}
\label{ss:VMP contin}

After fixing an arbitrary ordering on $L(C)$, each partial social rule on $\simp_\rats(L(C))$ of size $1$ can be identified with 
an arbitrary function on an arbitrary nonzero subset of the rational points of the $6$-simplex 
$\simp_6$, with image contained in $C$. Rules defined in this level of generality are not easy to deal with. Young \cite{Youn1975} introduced the axiom of continuity.
\begin{defn}
\label{def:continuous}
An anonymous rule $R$ is \hl{continuous} if when $E = (C, V, \pi), E' = (C, V', \pi')$ and $R(E) = \{r\}$ 
then $R(kE + E') = \{r\}$ for all sufficiently large integers $k$.
\end{defn}

If $R$ is homogeneous then it is continuous if and only if every vote distribution sufficiently close
to $E$ in the $\ell^1$-norm on $\simp(L(C))$ yields the same output as $E$. We do not know of any voting rule seriously considered in the literature that is not continuous. 

We now give a slight generalization of a result of Elkind, Faliszewski and Slinko \cite[Thm 6]{EFS2015}. 
\begin{prop}
\label{prop:contin VMP}
Suppose that $\cons$ is continuous and homogeneous, $d$ is votewise with respect to a continuous homogeneous seminorm and $(\cons,d)$  satisfies the VMP. Then $\R(\cons, d)$ is continuous.
\end{prop}
\begin{proof}
Let $E=(C, V, \pi), E'=(C, V', \pi') \in \elecs$ with $R(E) = \{r\}$. Thus there is $F=(C, V, \tau) \in \cons_r$ such that for all $r'\neq r$ and all $F'\in \cons_{r'}$,
$d(E, F) < d(E, F')$. Note that $d(E, F) = N(S)$ where $S$ is the multiset of all $d(\pi_i, \tau_i)$, and similarly $d(E, F') = N(S')$.

Fix $F''=(C, V', \pi'')\in \cons_r$. Then $d(kE+E', kF+F'') = N(kS+S'')$. Also for $d(kE+E', F')=N(kS'+T)$ for some $T$. By homogeneity of the norm and its continuity, we have for sufficiently large $k$, with $\varepsilon:=|V'|/k$,
$d(kE+E', kF+F'') = N(S+\varepsilon S'') < N(S'+\varepsilon T ).$
\end{proof}

\begin{cor}
\label{cor:}
If $1\leq p \leq \infty$, then $\R(\sunam^s, d^p)$ is continuous.
\end{cor}

\section{Conclusions and future work}
\label{s:conc}

We have clarified the relationship between distance rationalizability and axiomatic properties of social rules, and given improved necessary and sufficient conditions for rules to satisfy several of these axioms. The results show clearly that votewise distances combine better with votewise consensuses (which we define as those satisfying the VMP). The more complicated structure of consensuses such as $\cond^s$ compared to $\sunam^s$ is reflected in the failure of various properties to extend. Of course, VMP is a very strong property, and we
do not know of consensuses other than $\sunam^s$ that satisfy it 
generally. However VMP and CMP may be satisfied a particular $(\cons, d)$ pair in a given application.
What seems clear is that votewise distances work best with ``votewise consensuses", and Condorcet consensus with tournament distances. Mixing the two yields rules such as Dodgson's and the Voter Replacement Rule which fail to satisfy basic properties such as homogeneity.

We have only a few sufficient conditions for homogeneity of a DR rule. If the rule is not homogeneous, a homogeneous rule similar to the original may be found.
In \cite{Fish1977} a way around the nonhomogeneity of Dodgson's and Young's rules was
found, by using a limiting process to redefine the distance. This is unsatisfactory --- it is not even clear that the limit exists. Presumably using the construction  $\R(\overline{\cons}, \overline{d})$ may work, but it is not completely clear to us.

Systematic exploration of the space of rules $\R(\sunam^s, d^p)$ where $d$ is a neutral distance on rankings, may well unearth new rules with desirable properties. These rules are already known to be continuous, neutral, anonymous, homogeneous and consistent. Other possibly desirable properties may also be satisfied: Kemeny's rule, which falls into this class, also 
satisfies a Condorcet property for social welfare rules \cite{YoLe1978} while scoring rules are also monotonic as social choice rules.

To our knowledge, $\ell^p$ votewise distances with $1 < p < \infty$ have not been 
studied systematically. Also, in addition to the discrete, inversion and Spearman metrics on rankings  discussed here, there are many interesting distances on rankings yet to be explored. Besides votewise distances, there are many other interesting distances on $\votesits$ and $\votedists$, which may yield useful new social rules. Such \hl{distances on multisets} and \hl{statistical distances} are heavily used in many application areas \cite{DeDe2009}.

In a related forthcoming work, the present authors use the framework of distance rationalizability of anonymous and homogeneous rules to study the decisiveness of such rules. We expect other applications, for example by using different groups of symmetries.

\printbibliography

\end{document}